\documentclass[aps,superscriptaddress,twocolumn,pra,showpacs]{revtex4-1}

\usepackage{amsmath}
\usepackage{amssymb}
\usepackage{epsfig}
\usepackage{color}
\usepackage{amsthm}
\usepackage{hyperref}
\usepackage{appendix}

\newtheorem{theorem}{Theorem}
\newtheorem{proposition}{Proposition}

\newtheorem{corollary}{Corollary}

\newcommand{\ket}[1]{|#1\rangle}
\newcommand{\bra}[1]{\langle #1|}
\newcommand{\bracket}[2]{\langle #1|#2\rangle}
\newcommand{\ketbra}[2]{|#1\rangle\langle #2|}
\newcommand{\tr}[0]{\textnormal{Tr}}

\begin{document}
\title{Conclusive exclusion of quantum states}
\author{Somshubhro Bandyopadhyay}
\affiliation{Department of Physics and Center for Astroparticle Physics and Space Science, Bose Institute, Block EN, Sector V, Bidhan Nagar, Kolkata 700091, India}
\author{Rahul Jain}
\affiliation{Department of Computer Science and Centre for Quantum Technologies, National University of Singapore, Singapore 119615}
\author{Jonathan Oppenheim}
\affiliation{Department of Computer Science and Centre for Quantum Technologies, National University of Singapore, Singapore 119615}
\affiliation{Department of Physics and Astronomy, University College London, Gower Street, London WC1E 6BT, United Kingdom}
\author{Christopher Perry}
\affiliation{Department of Physics and Astronomy, University College London, Gower Street, London WC1E 6BT, United Kingdom}

\begin{abstract}
In the task of quantum state exclusion we consider a quantum system, prepared in a state chosen from a known set. The aim is to perform a measurement on the system which can conclusively rule that a subset of the possible preparation procedures can not have taken place. We ask what conditions the set of states must obey in order for this to be possible and how well we can complete the task when it is not. The task of quantum state discrimination forms a subclass of this set of problems. Within this paper we formulate the general problem as a Semidefinite Program (SDP), enabling us to derive sufficient and necessary conditions for a measurement to be optimal. Furthermore, we obtain a necessary condition on the set of states for exclusion to be achievable with certainty and give a construction for a lower bound on the probability of error. This task of conclusively excluding states has gained importance in the context of the foundations of quantum mechanics due to a result of Pusey, Barrett and Rudolph (PBR). Motivated by this, we use our SDP to derive a bound on how well a class of hidden variable models can perform at a particular task, proving an analogue of Tsirelson's bound for the PBR experiment and the optimality of a measurement given by PBR in the process. We also introduce variations of conclusive exclusion, including unambiguous state exclusion, and state exclusion with worst case error.
\end{abstract}

\maketitle

\section{Introduction}
Suppose we are given a single-shot device, guaranteed to prepare a system in a quantum state chosen at random from a finite set of $k$ known states. In the  quantum state discrimination problem, we would attempt to identify the state that has been prepared. It is a well known result \cite{Nielsen2010} that this can be done with certainty if and only if all of the states in the set of preparations are orthogonal to one another. By allowing inconclusive measurement outcomes \cite{Ivanovic1987, Dieks1988, Peres1988} or accepting some error probability \cite{Helstrom1976,Holevo1973, Yuen1975}, strategies can be devised to tackle the problem of discriminating between non-orthogonal states. For a recent review of quantum state discrimination, see \cite{Barnett2009}. What however, can we deduce about the prepared state with certainty?

Through state discrimination we effectively attempt to increase our knowledge of the system so that we progress from knowing it is one of $k$ possibilities to knowing it is one particular state. We reduce the size of the set of possible preparations that could have occurred from $k$ to $1$. A related, and less ambitious task, would be to exclude $m$ preparations from the set, reducing the size of the set of potential states from $k$ to $k-m$. If we rule out the $m$ states with certainty we say that they have been conclusively excluded. Conclusive exclusion of a single state is not only interesting from the point of view of the theory of measurement, but it is becoming
increasingly important in the foundations of quantum theory. It has previously been considered with respect to quantum state compatibility criteria between three parties \cite{Caves2002} where Caves et al. derive necessary and sufficient conditions for conclusive exclusion of a single state from a set of three pure states to be possible. More recently it has found use in investigating the plausibility of $\psi$-epistemic theories describing quantum mechanics \cite{Pusey2012}.

As recognized in \cite{Pusey2012} for the case of single state exclusion, the problem of conclusive exclusion can be formulated in the framework of Semidefinite Programs (SDPs). As well as being efficiently numerically solvable, SDPs also offer a structure that can be exploited to derive statements about the underlying problem they describe \cite{Vandenberghe1996, Watrous2011}. This has already been applied to the problem of state discrimination \cite{Jezek2002, Eldar2003b, Eldar2003}. Given that minimum error state discrimination forms a subclass ($m=k-1$) of the general exclusion framework, it is reasonable to expect that a similar approach will pay dividends here.

For minimum error state discrimination, SDPs provide a route to produce necessary and sufficient conditions for a measurement to be optimal. Similarly, the SDP formalism can be applied to obtain such conditions for the task of minimum error state exclusion and we derive these in this paper. By applying these requirements to exclusion problems, we have a method for proving whether a given measurement is optimal for a given ensemble of states.

From the SDP formalism it is also possible to derive necessary conditions for $m$-state conclusive exclusion to be possible for a given set of states and lower bounds on the probability of error when it is not. A special case of this result is the fact that state discrimination can not be achieved when the set of states under consideration are non-orthogonal. By regarding perfect state discrimination as $(k-1)$-state conclusive exclusion, we re-derive this result.

As an application of our SDP and its properties we consider a game, motivated by the argument, due to PBR \cite{Pusey2012}, against a class of hidden variable theories. Assume that we have a physical theory, not necessarily that of quantum mechanics, such that, when we prepare a system, we describe it by a state, $\chi$. If our theory were quantum mechanics, then $\chi$ would be identified with $\ket{\psi}$, the usual quantum state. Furthermore, suppose that $\chi$ does not give a complete description of the system. We assume that such a description exists, although it may always be unknown to us, and denote it $\lambda$. As $\chi$ is an incomplete description of the system, it will be compatible with many different complete states. We denote these states $\lambda\in\Lambda_{\chi}$. PBR investigate whether for distinct quantum descriptions, $\ket{\psi_0}$ and $\ket{\psi_1}$, it is possible that $\Lambda_{\ket{\psi_0}}\cap\Lambda_{\ket{\psi_1}}\neq\emptyset$. Models that satisfy this criteria are called $\psi$-epistemic, see \cite{Harrigan2010} for a full description.

Consider now the following scenario. Alice gives Bob a system prepared according to one of two descriptions, $\chi_1$ or $\chi_2$, and Bob's task is to identify which preparation he has been given. Bob observes the system and will identify the wrong preparation with probability $q$. Note that $0\leq q\leq 1/2$, as Bob will always have the option of randomly guessing the description without performing an observation. If $\Lambda_{\chi_1}\cap\Lambda_{\chi_2}\neq\emptyset$ then, even if Bob has access to the complete description of the system, $\lambda$, $q>0$ as there will exist $\lambda$ compatible with both $\chi_1$ and $\chi_2$.

Now suppose Bob is given $n$ such systems prepared independently, and we represent the preparation as a string in $\{0,1\}^n$. Bob's task is to output such an $n$-bit string and he wins if his is not identical to the string corresponding to Alice's preparation, i.e., he attempts to exclude one of the $2^n$ preparations. We refer to this as the `PBR game' and we will consider two scenarios for playing it. Under the first scenario, Bob can only perform measurements on each system individually. We refer to this as the separable version of the game. In the second scenario we allow Bob to perform global measurements on the $n$ systems he receives. We refer to this as the global version, and we are interested in how well quantum theory performs in this case.  We shall make a key assumption of PBR: that the global complete state of $n$ independent systems, $\Omega$, is given by the tensor product of the individual systems' complete states. This second, quantum, task is related to the problem of `Hedging bets with correlated quantum strategies' as introduced in \cite{Molina2012} and expanded upon in \cite{Arunachalam2013}.

By calculating Bob's probability of success in the PBR game under each of these schemes we gain a measure of how the predictions of quantum mechanics compare with the predictions of theories in which both $\Lambda_{\chi_1}\cap\Lambda_{\chi_2}\neq\emptyset$ and $\Omega=\otimes_{i=1}^{n}\lambda_i$ hold. As such, the result can be seen as similar in spirit to Tsirelson's bound \cite{Tsirelson1980} in describing how well quantum mechanical strategies can perform at the CHSH game.

This paper is organized as follows. First, in Section \ref{Section: The State Exclusion SDP}, we formulate the quantum state exclusion problem as an SDP, developing the structure we will need to analyze the task. Next, in Section \ref{Section: The optimal exclusion measurement}, we derive sufficient and necessary conditions for a measurement to be optimal in performing conclusive exclusion. It is these conditions that will assist us in investigating the entangled version of the PBR game. In Section \ref{Section: Necessary condition} we derive a necessary condition on the set of possible states for single state exclusion to be possible and in Section \ref{Section: Lower Bound} we give a lower bound on the probability of error when it is not. We apply the SDP formalism to the PBR game in Section \ref{Section: The PBR game} and use it to quantify the discrepancy between the predictions of a class of hidden variable theories and those of quantum mechanics. Finally, in Section \ref{Section: Alternative measures of exclusion}, we present alternative formulations of state exclusion and construct the relevant SDPs.

\section{The State Exclusion SDP} \label{Section: The State Exclusion SDP}
More formally, what does it mean to be able to perform conclusive exclusion? We first consider the case of single state exclusion and then show how it generalizes to $m$-state exclusion. Let the set of possible preparations on a $d$ dimensional quantum system be $\mathcal{P}=\left\{\rho_{i}\right\}^{k}_{i=1}$ and let each preparation occur with probability $p_i$. For brevity of notation we define $\tilde{\rho}_i=p_i \rho_i$. Call the prepared state $\sigma$. The aim is to perform a measurement on $\sigma$ so that, from the outcome, we can state $j\in\{1,\ldots,k\}$ such that $\sigma\neq\rho_j$.

Such a measurement will consist of $k$ measurement operators, one for attempting to exclude each element of $\mathcal{P}$. We want a measurement, described by $\mathcal{M}=\{M_i\}_{i=1}^{k}$, that never leads us to guess $j$ when $\sigma=\rho_j$. We need:
\begin{equation}
\tr\left[\rho_i M_i\right]=0, \quad \forall i, \label{single exclusion}
\end{equation}
or equivalently, since $\rho_i$ and $M_i$ are positive semidefinite matrices and $p_i$ is a positive number:
\begin{equation}
\alpha=\sum_{i=1}^{k}\tr\left[\tilde{\rho}_i M_i\right]=0. \label{excl cond}
\end{equation}
There will be some instances of $\mathcal{P}$ for which a $\mathcal{M}$ can not be found to satisfy Eq. (\ref{excl cond}). In these cases our goal is to minimize $\alpha$ which corresponds to the probability of failure of the strategy, `If outcome $j$ occurs say $\sigma\neq\rho_j$'.

Therefore, to obtain the optimal strategy for single state exclusion, our goal is to minimize $\alpha$ over all possible $\mathcal{M}$ subject to $\mathcal{M}$ forming a valid measurement. Such an optimization problem can be formulated as an SDP:
\begin{align}
\begin{split}
\underset{\mathcal{M}}{\textrm{Minimize: }}&\alpha=\sum_{i=1}^{k}\tr\left[\tilde{\rho}_i M_i\right]. \\
\textrm{Subject to: }&\sum_{i=1}^{k}M_i=\mathbb{I},\\
& M_i\geq0, \quad \forall i. \label{SDP Primal}
\end{split}
\end{align}
Here $\mathbb{I}$ is the $d$ by $d$ identity matrix and $A\geq 0$ implies that $A$ is a positive semidefinite matrix. The constraint $\sum_{i=1}^{k}M_i=\mathbb{I}$ corresponds to the fact that the $M_i$ form a complete measurement
and we don't allow inconclusive results.

Part of the power of the SDP formalism lies in constructing a `dual' problem to this `primal' problem given in Eq. (\ref{SDP Primal}). Details on the formation of the dual problem to the exclusion SDP can be found in
Appendix \ref{State Exclusion SDP Formulation}
and we state it here:

\begin{align}
\begin{split}
\underset{N}{\textrm{Maximize: }}&\beta=\tr\left[N\right].\\
\textrm{Subject to: }&N\leq \tilde{\rho}_i, \quad\forall i,\\
&N\in \textrm{Herm}. \label{SDP Dual}
\end{split}
\end{align}
For single state exclusion, the problem is essentially to maximize the trace of a Hermitian matrix $N$ subject to $\tilde{\rho}_i-N$ being a positive semidefinite matrix, $\forall$ $i$.

What of $m$-state conclusive exclusion? Define $Y_{(k,m)}$ to be the set of all subsets of the integers $\{1,\ldots,k\}$ of size $m$. The aim is to perform a measurement on $\sigma$ such that from the outcome we can state a set, $Y\in Y_{(k,m)}$, such that $\sigma\notin\{\rho_y\}_{y\in Y}$. Such a measurement, denoted $\mathcal{M}_m$, will consist of ${k \choose m}$ measurement operators and we require that, for each set $Y$:
\begin{equation}
\tr\left[\tilde{\rho}_y M_Y\right]=0,\quad \forall y\in Y.
\end{equation}
If we define:
\begin{equation}
\hat{\rho}_Y = \sum_{y\in Y} \tilde{\rho}_y,
\end{equation}
then this can be reformulated as requiring:
\begin{equation}
\tr\left[\hat{\rho}_Y M_Y\right]=0, \quad \forall Y\in Y_{(k,m)}. \label{m exclusion}
\end{equation}
Eq. (\ref{m exclusion}) is identical in form to Eq. (\ref{single exclusion}). Hence we can view $m$-state exclusion as single state exclusion on the set $\mathcal{P}_m=\{\hat{\rho}_Y\}_{Y\in Y_{(k,m)}}$. Furthermore, we can generalize this approach to an arbitrary collection of subsets that are not necessarily of the same size. With this in mind we restrict ourselves to considering single state exclusion in all that follows.

The tasks of state exclusion and state discrimination share many similarities. Indeed, if we instead maximize $\alpha$ in Eq. (\ref{SDP Primal}) and minimize $\beta$ in Eq. (\ref{SDP Dual}) together with inverting the inequality constraint to read $N\geq\tilde{\rho}_i$, we obtain the SDP associated with minimum error state discrimination. It is also possible to recast each problem as an instance of the other. Firstly, state discrimination can be put in the form of an exclusion problem by taking $m=k-1$ because if we exclude $k-1$ of the possible states, then we can identify $\sigma$ as the remaining state.

Following the observation of \cite{Nakahira2012} regarding minimum Bayes cost problems, state exclusion can be converted into a discrimination task. To see this, from $\mathcal{P}$ define:
\begin{equation}
\mathcal{R}=\left\{\vartheta_i=\frac{1}{k-1}\sum_{j\neq i}\tilde{\rho}_j\right\}_{i=1}^{k}.
\end{equation}
Writing $P_{\textit{error}}^{\textit{dis}}$ and $P_{\textit{error}}^{\textit{exc}}$ to distinguish between the probability of error in discrimination and exclusion, in state discrimination on $\mathcal{R}$ we would attempt to minimize:
\begin{align}
P_{\textit{error}}^{\textit{dis}}\left(\mathcal{R}\right)&=1-\sum_{i=1}^{k}\tr\left[\vartheta_i M_i\right]
\intertext{which can be rearranged to give (see Appendix \ref{exc-disc derivation}):}
P_{\textit{error}}^{\textit{dis}}\left(\mathcal{R}\right)&=\frac{k-2}{k-1}+\frac{1}{k-1}P_{\textit{error}}^{\textit{exc}}\left(\mathcal{P}\right). \label{exc-disc}
\end{align}
Hence, minimizing the error probability in discrimination on $\mathcal{R}$ is equivalent to minimizing the probability of error in state exclusion on $\mathcal{P}$ and the optimal measurement is the same for both. This interplay between the two tasks enables us to apply bounds on the error probability of state discrimination (see for example \cite{Qiu2010}) to the task of state exclusion.

Returning to the SDP, let us define the optimum solution to the primal problem to be $\alpha^{*}$ and the solution to the corresponding dual to be $\beta^{*}$. It is a property of all SDPs, known as weak duality, that $\beta\leq\alpha$. Furthermore, for SDPs satisfying certain conditions, $\alpha^{*}=\beta^{*}$ and this is known as strong duality. The exclusion SDP does fulfill these criteria, as shown in
 Appendix \ref{Slater's Theorem applied to Exclusion SDP}.
 Using weak and strong duality allows us to derive properties of the optimal measurement for the problem, a necessary condition on $\mathcal{P}$ for conclusive exclusion to be possible and a bound on the probability of error in performing the task.

\section{The optimal exclusion measurement} \label{Section: The optimal exclusion measurement}
Strong duality gives us a method for proving whether a feasible solution, satisfying the constraints of the primal problem, is an optimal solution. If $\mathcal{M}^{*}$ is an optimal measurement for the conclusive exclusion SDP, then, by strong duality, there must exist a Hermitian matrix $N^{*}$, satisfying the constraints of the dual problem, such that:
\begin{equation}
\sum_{i=1}^{k} \tr\left[\tilde{\rho}_i M_{i}^{*}\right]=\tr\left[N^{*}\right].
\end{equation}
Furthermore, the following is true:
\begin{theorem} \label{SD theorem}
Suppose a state $\sigma$ is prepared at random using a preparation from the set $\mathcal{P}$ according to some probability distribution $\{p_i\}_{i=1}^{k}$. Applying the measurement $\mathcal{M}$ to $\sigma$ is optimal for attempting to exclude a single element from the set of possible preparations if and only if:
\begin{equation}
N=\sum_{i=1}^{k} \left[\tilde{\rho}_i M_i\right], \label{N construct}
\end{equation}
is Hermitian and satisfies $N\leq \tilde{\rho}_{i}$, $\forall i$.
\end{theorem}
The proof of Theorem \ref{SD theorem} is given in 
Appendix \ref{Necessary and sufficient conditions for a measurement to be optimal} 
 and revolves around the application of strong duality together with a property called complementary slackness. It is similar in construction to Yuen et al.'s \cite{Yuen1975} derivation of necessary and sufficient conditions for showing a quantum measurement is optimal for minimizing a given Bayesian cost function.
This result provides us with a method for proving a measurement is optimal; we construct $N$ according to Eq. (\ref{N construct}) and show that it satisfies the constraints of the dual problem. It is this technique which will allow us to analyze the PBR game in the quantum setting.

\section{Necessary condition for single state conclusive exclusion} \label{Section: Necessary condition}
Through the application of weak duality we can also gain insight into the SDP. As the optimal solution to the dual problem provides a lower bound on the solution of the primal problem, any feasible solution to the dual does too although it may not necessarily be tight. This relation can be summarized as:
\begin{equation}
\tr\left[N^{feas}\right]\leq\tr\left[N^{*}\right]=\beta^{*}=\alpha^{*}.
\end{equation}
In particular if, for a given $\mathcal{P}$, we can construct a feasible $N$ with $\tr\left[N\right]>0$, then we have $\alpha^{*}>0$ and hence conclusive exclusion is not possible.

Constructing such an $N$ gives rise to the following necessary condition on the set $\mathcal{P}$ for conclusive exclusion to be possible: 
\begin{theorem} \label{Necc Cond Theorem}
Suppose a system is prepared in the state $\sigma$ using a preparation chosen at random from the set $\mathcal{P}=\{\rho_i\}_{i=1}^{k}$. Single state conclusive exclusion is possible only if:
\begin{equation}
\sum_{j\neq l=1}^{k} F(\rho_j,\rho_l)\leq k(k-2), \label{Necc Con}
\end{equation}
where $F(\rho_j,\rho_l)$ is the fidelity between states $\rho_j$ and $\rho_l$.
\end{theorem}
The full proof of this theorem is given in 
Appendix \ref{Necessary Condition for Conclusive Exclusion}. 
but we sketch it here.
Define $N$ as follows:
\begin{equation}
N=-p\sum_{r=1}^{k}\rho_r +\frac{1-\epsilon}{k-2}p\sum_{1\leq j<l\leq k}\left(\sqrt{\rho_j}U_{jl}\sqrt{\rho_l}+\sqrt{\rho_l}U_{jl}^{*}\sqrt{\rho_j}\right),
\end{equation}
where the $U_{jl}$ are unitary matrices chosen such that:
\begin{equation}
\tr\left[N\right]=-kp+\frac{1-\epsilon}{k-2} p\sum_{j\neq l=1}^{k} F(\rho_j,\rho_l).
\end{equation}
$N$ is Hermitian and for suitable $p$ and $\epsilon$ it can be shown that $\rho_i-N\geq 0$, $\forall i$. Eq. (\ref{Necc Con}) follows by determining when $\tr[N]>0$ and letting $\epsilon \rightarrow 0$.
Note that the probability with which states are prepared, $\{p_i\}_{i=1}^{k}$, does not impact on whether conclusive exclusion is possible or not.

This is only a necessary condition for single state conclusive exclusion and there exist sets of states that satisfy Eq. (\ref{Necc Con}) for which it is not possible to perform conclusive exclusion. Nevertheless, there exist sets of states on the cusp of satisfying Eq. (\ref{Necc Con}) for which conclusive exclusion is possible. For example, the set of states of the form:
\begin{equation}
\ket{{\psi}_i}=\sum_{j\neq i}^{k}\frac{1}{\sqrt{k-1}}\ket{j},
\end{equation}
for $i=1$ to $k$, can be conclusively excluded by the measurement in the orthonormal basis $\{\ket{i}\}_{i=1}^{k}$ and yet:
\begin{align}
\begin{split}
\sum_{j\neq l=1}^{k} F\left(\ketbra{\psi_j}{\psi_j},\ketbra{\psi_l}{\psi_l}\right)&=\sum_{j\neq l=1}^{k}|\bracket{\psi_j}{\psi_l}|\\
&= k(k-2).
\end{split}
\end{align}


It can be shown that the necessary condition for conclusive state discrimination can be obtained from Theorem \ref{Necc Cond Theorem} and the interested reader can find this derivation in Appendix \ref{Necessary condition for conclusive state discrimination}.

\section{Lower bound on the probability of error} \label{Section: Lower Bound}
Weak duality can also be used to obtain the following lower bound on $\alpha^{*}$:
\begin{theorem} \label{lower bound}
For two Hermitian operators, $A$ and $B$, define $\min\left(A,B\right)$ to be:
\begin{equation}
\min\left(A,B\right)=\frac{1}{2}\left[A+B-|A-B|\right].
\end{equation}
Given a set of states $\mathcal{P}=\{\rho_i\}_{i=1}^{k}$ prepared according to some probability distribution $\{p_i\}_{i=1}^{k}$ and a permutation $\varepsilon$, acting on $k$ objects, taken from the permutation group $S_k$, consider:
\begin{equation}
N_{\varepsilon}=\min\left(\tilde{\rho}_{\varepsilon(k)},\min\left(\tilde{\rho}_{\varepsilon(k-1)},\min\left(\ldots,\min\left(\tilde{\rho}_{\varepsilon(2)},\tilde{\rho}_{\varepsilon(1)}\right)\right)\right)\right). \label{Nepsilon}
\end{equation}
Then:
\begin{equation}
\alpha^{*}\geq\underset{\varepsilon\in S_k}{\textnormal{Maximum: }}\tr\left[N_{\varepsilon}\right].
\end{equation}
\end{theorem}
The proof of this result is given in Appendix \ref{lower bound proof} and relies upon showing that $\min(A,B)\leq A \textrm{ and }B$, together with the iterative nature of the construction of $N_{\varepsilon}$. Note that by considering a suitably defined $\max$ function, analogous to the $\min$ used in Theorem \ref{lower bound}, it is possible to derive a similar style of bound for the task of minimum error state discrimination. We omit it here however, as it is beyond the scope of this paper.

\section{The PBR game} \label{Section: The PBR game}
We now turn our attention to the PBR game. Suppose Alice gives Bob $n$ systems whose preparations are encoded by the string $\vec{x}\in\{0,1\}^n$. The state of system $i$ is $\chi_{x_i}$. Bob's goal is to produce a string $\vec{y}\in\{0,1\}^n$ such that $\vec{x}\neq\vec{y}$.

\subsection{Separable version}
In the first scenario, where Bob can only observe each system individually and we consider a general theory, we can represent his knowledge of the global system by:
\begin{equation}
\Gamma=\gamma_1\otimes\ldots\otimes\gamma_n,
\end{equation}
with $\gamma_i\in\{\Gamma_0,\Gamma_1,\Gamma_?\}$, representing his three possible observation outcomes. If $\gamma_i\in\Gamma_0$ he is certain the system preparation is described by $\chi_0$, if $\gamma_i\in\Gamma_1$ he is certain the system preparation is described by $\chi_1$ and if $\gamma_i\in\Gamma_?$ he remains uncertain whether the system was prepared in state $\chi_0$ or $\chi_1$ and he may make an error in assigning a preparation to the system. We denote the probability that Bob, after performing his observation, assigns the wrong preparation description to the system as $q$. Provided that $\Gamma_?\neq\emptyset$, then $q>0$.

Bob will win the game if for at least one individual system he assigns the correct preparation description. His strategy is to attempt to identify each value of $x_i$ and choose $y_i$ such that $y_i\neq x_i$. Bob's probability of outputting a winning string is hence:
\begin{equation}
P_{\textit{win}}^{S}=1-q^n. \label{P win sep}
\end{equation}

\subsection{Global version}
Now consider the second scenario and when the theory is quantum and global (i.e., entangled) measurements on the global system are allowed. We can write the global state that Alice gives Bob, labeled by $\vec{x}$, as:
\begin{equation}
\ket{\Psi_{\vec{x}}}=\bigotimes^{n}_{i=1}\ket{\psi_{x_i}}.
\end{equation}
Bob's task can now be regarded as attempting to perform single state conclusive exclusion on the set of states $\mathcal{P}=\{\ket{\Psi_{\vec{x}}}\}_{\vec{x}\in\{0,1\}^n}$; he outputs the string associated to the state he has excluded to have the best possible chance of winning the game.

To calculate his probability of winning $P_{\textit{win}}^{G}$ we need to construct and solve the associated SDP. Without loss of generality, we can take the states $\ket{\psi_0}$ and $\ket{\psi_1}$ to be defined as:
\begin{align}
\begin{split}
\ket{\psi_0}&=\cos\left(\frac{\theta}{2}\right)\ket{0}+\sin\left(\frac{\theta}{2}\right)\ket{1}, \\
\ket{\psi_1}&=\cos\left(\frac{\theta}{2}\right)\ket{0}-\sin\left(\frac{\theta}{2}\right)\ket{1},
\end{split}
\end{align}
where $0\leq\theta\leq\pi/2$. The global states $\ket{\Psi_{\vec{x}}}$ are then given by:
\begin{equation}
\ket{\Psi_{\vec{x}}}=\sum_{\vec{r}}\left(-1\right)^{\vec{x}\cdot\vec{r}}\left[\cos\left(\frac{\theta}{2}\right)\right]^{n-|\vec{r}|}\left[\sin\left(\frac{\theta}{2}\right)\right]^{|\vec{r}|}\ket{\vec{r}},
\end{equation}
where $\vec{r}\in\{0,1\}^n$ and $|\vec{r}|=\sum^{n}_{i=1}r_i$.

From \cite{Pusey2012}, we know that single state conclusive exclusion can be performed on this set of states provided $\theta$ and $n$ satisfy the condition:
\begin{equation}
2^{1/n}-1\leq \tan\left(\frac{\theta}{2}\right). \label{PBR Criterion}
\end{equation}
When this relation holds, $P_{\textit{win}}^{G}=1$. What however, happens outside of this range? Whilst strong numerical evidence is given in \cite{Pusey2012} that it will be the case that $P_{\textit{win}}^{G}<1$, can it be shown analytically?

Through analyzing numerical solutions to the SDP (performed using \cite{Lofberg2004}, \cite{Sturm1999}), there is evidence to suggest that the optimum measurement to perform when Eq. (\ref{PBR Criterion}) is not satisfied is given by the projectors:
\begin{align}
\ket{\zeta_{\vec{x}}}=\frac{1}{\sqrt{2^n}}\left(\ket{\vec{0}}-\sum_{\vec{r}\neq\vec{0}}\left(-1\right)^{\vec{x}\cdot\vec{r}}\ket{\vec{r}}\right), \label{PBR Projector}
\end{align}
which are independent of $\theta$. That the set $\{\ket{\zeta_{\vec{x}}}\}_{\vec{x}\in\{0,1\}^n}$ is the optimal measurement for attempting to perform conclusive exclusion is shown in 
Appendix \ref{PBR Game}.

If we construct $N$ as per Eq. (\ref{N construct}) and consider the trace, we can determine how successfully single state exclusion can be performed. This is done in 
Appendix \ref{PBR Game} 
and we find:
\begin{equation}
\tr\left[N\right]=\frac{1}{2^n}\left[\cos\left(\frac{\theta}{2}\right)\right]^{2n}\left(2-\left[1+\tan\left(\frac{\theta}{2}\right)\right]^n\right)^2.
\end{equation}
This is strictly positive and hence we have shown that Eq. (\ref{PBR Criterion}) is a necessary condition for conclusive exclusion to be possible on the set $\mathcal{P}$.

In summary, we have:
\begin{align}
\begin{split}
&\textrm{If: } 2^{1/n}-1\leq \tan\left(\frac{\theta}{2}\right), \\
&\quad P_{\textit{win}}^{G}=1. \\
&\textrm{Else:} \\
&\quad P_{\textit{win}}^{G}=1-\frac{1}{2^n}\left[\cos\left(\frac{\theta}{2}\right)\right]^{2n}\left(2-\left[1+\tan\left(\frac{\theta}{2}\right)\right]^n\right)^2
\end{split}
\end{align}
which characterizes the success probability of the quantum strategy.

\subsection{Comparison}
What is the relation between $P_{\textit{win}}^{S}$ and $P_{\textit{win}}^{G}$? If, in the separable scenario, we take the physical theory as being quantum mechanics and Bob's error probability as arising from the fact that it is impossible to distinguish between non-orthogonal quantum states, we can write \cite{Helstrom1976}:
\begin{align}
\begin{split}
q&=\left(\frac{1}{2}\right)\left(1-\sqrt{1-\left|\bracket{\psi_0}{\psi_1}\right|^2}\right)\\
&=\left(\frac{1}{2}\right)\left(1-\sin\left(\theta\right)\right).
\end{split}
\end{align}
With this substitution we find that $P_{\textit{win}}^{S}\leq P_{\textit{win}}^{G}$, $\forall n$. This is unsurprising as the first scenario is essentially the second but with a restricted set of allowable measurements.

Of more interest however, is if we view $q$ as arising from some hidden variable completion of quantum mechanics. If $\Lambda_{\ket{\psi_0}}\cap\Lambda_{\ket{\psi_1}}=\emptyset$, then if an observation of each $\ket{\psi_{x_i}}$ were to allow us to deduce $\lambda_{x_i}$ then $q=0$ and $P_{\textit{win}}^{S}=1\geq P_{\textit{win}}^{G}$. However, if $\Lambda_{\ket{\psi_0}}\cap\Lambda_{\ket{\psi_1}}\neq\emptyset$, then we have $q>0$ and $P_{\textit{win}}^{S}$ will have the property that Bob wins with certainty only as $n\rightarrow\infty$. On the other hand, $P_{\textit{win}}^{G}=1$ if and only if Eq. (\ref{PBR Criterion}) is satisfied and we have analytically proven the necessity of the bound obtained by PBR. Furthermore, we have defined a game that allows the quantification of the difference between the predictions of general physical theories, including those that attempt to provide a more complete description of quantum mechanics, and those of quantum mechanics.

\section{Alternative measures of exclusion} \label{Section: Alternative measures of exclusion}

There exist multiple strategies and figures of merit when undertaking state discrimination. In addition to considering minimum error discrimination or unambiguous discrimination, further variants may try to minimize the maximum error probability \cite{Kosut2004} or allow only a certain probability of obtaining an inconclusive measurement result \cite{Fiuravsek2003}. Similarly, alternative methods to minimum error can be defined for state exclusion and in this section unambiguous exclusion and worst case error exclusion are defined and the related SDPs given.

\subsection{Unambiguous State Exclusion}

In unambiguous state exclusion on the set of preparations $\mathcal{P}=\{\tilde{\rho}_i\}_{i=1}^{k}$ we consider a measurement given by $\mathcal{M}=\{M_1, \ldots, M_k, M_?\}$. If we obtain measurement outcome $i$ $(1\leq i\leq k)$, then we can exclude with certainty the state $\rho_i$. However, if we obtain the outcome labeled $?$, we can not infer which state to exclude. We wish to minimize the probability of obtaining this inconclusive measurement:
\begin{equation}
\alpha=\sum_{i=1}^{k}\tr\left[\tilde{\rho}_i M_?\right],
\end{equation}
which can be rewritten as:
\begin{equation}
\alpha=\tr\left[\sum_{j=1}^{k}\tilde{\rho}_j\left(\mathbb{I}-\sum_{i=1}^{k}M_i\right)\right].
\end{equation}

Defining $\tilde{\alpha}=1-\alpha$, the primal SDP associated with this task is given by:
\begin{align}
\begin{split}
\underset{\mathcal{M}}{\textrm{Maximize: }}&\tilde{\alpha}=\tr\left[\sum_{j=1}^{k}\tilde{\rho}_j\sum_{i=1}^{k}M_i\right]. \\
\textrm{Subject to: }&\sum_{i=1}^{k}M_i\leq\mathbb{I},\\
& \tr\left[\tilde{\rho_i}M_i\right]=0, \quad 1\leq i\leq k, \\
& M_i\geq0, \quad 1\leq i\leq k. \label{Unambig Prime}
\end{split}
\end{align}
Here, the first and third constraints ensure that $\mathcal{M}$ is a valid measurement whilst the second, $\tr\left[\tilde{\rho}_i M_i\right]=0$, $1\leq i\leq k$, encapsulates the fact that when measurement outcome $i$ occurs we should be able to exclude state $\rho_i$ with certainty.

The dual problem can be shown to be (see Appendix \ref{Unambiguous State Exclusion SDP}):
\begin{align}
\begin{split}
\underset{N, \{a_i\}_{i=1}^{k}}{\textrm{Minimize: }}&\beta=\tr\left[N\right].\\
\textrm{Subject to: }&a_i\tilde{\rho}_i+N\geq \sum_{j=1}^{k}\tilde{\rho}_j, \quad 1\leq i\leq k,\\
&a_i\in\mathbb{R}, \quad \forall i,\\
&N\geq 0. \label{Unambig Dual}
\end{split}
\end{align}

Unambiguous state exclusion has recently found use in implementations of quantum digital signatures \cite{Collins2013}, enabling such schemes to be put into practice without the need for long term quantum memory. 

\subsection{Worst Case Error State Exclusion}

The goal of the SDP given in Eqs. (\ref{SDP Primal}) and (\ref{SDP Dual}) is to minimize the average probability of error, over all possible preparations, of the strategy, `If outcome $j$ occurs say $\sigma\neq\rho_j$'. An alternative goal would be to minimize the worst case probability of error that occurs:
\begin{equation}
\alpha=\max_{i}\tr\left[\tilde{\rho}_i M_i\right].
\end{equation}

The primal SDP associated with this task is:
\begin{align}
\begin{split}
\underset{\mathcal{M}}{\textrm{Minimize: }}&\alpha=\lambda. \\
\textrm{Subject to: }&\lambda\geq\tr\left[\tilde{\rho}_i M_i\right], \quad \forall i,\\
&\sum_{i=1}^{k}M_i=\mathbb{I},\\
& \lambda\geq0\in\mathbb{R}, \\
& M_i\geq0, \quad 1\leq i\leq k. \label{Worst Case Prime}
\end{split}
\end{align}
These constraints again encode that $\mathcal{M}$ forms a valid measurement and ensure that $\alpha$ picks out the worst case error probability across all possible preparations.

The associated dual problem is:
\begin{align}
\begin{split}
\underset{N, \{a_i\}_{i=1}^{k}}{\textrm{Maximize: }}&\beta=\tr\left[N\right].\\
\textrm{Subject to: }&N\leq a_i\tilde{\rho}_i, \quad \forall i,\\
&\sum_{i=1}^{k} a_i\leq1,\\
&a_i\geq0\in\mathbb{R}, \quad \forall i,\\
&N\in \textrm{Herm}. \label{Worst Case Dual}
\end{split}
\end{align}
The derivation of this is given in Appendix \ref{Worst Case Error State Exclusion SDP}.

\section{Conclusion} 
In this paper we have introduced the task of state exclusion and shown how it can be formulated as an SDP. Using this we have derived conditions for measurements to be optimal at minimum error state exclusion and a criteria for the task to be performed conclusively on a given set of states. We also gave a lower bound on the error probability. Furthermore, we have applied our SDP to a game which helps to quantify the differences between quantum mechanics and a class of hidden variable theories.

It is an open question, posed in \cite{Caves2002}, whether a POVM ever outperforms a projective measurement in conclusive exclusion of a single pure state. Whilst it can be shown from the SDP formalism that this is not the case when the states are linearly independent and conclusive exclusion is not possible to the extent that $\tr\left[M_i\rho_i\right]>0$, $\forall i$, further work is required to extend it and answer the above question. It would also be interesting to see whether it is possible to find further constraints and bounds, similar to Theorem \ref{Necc Cond Theorem} and Theorem \ref{lower bound}, to characterize when conclusive exclusion is possible.

Finally, the main SDP, as given in Eq. (\ref{SDP Primal}), is just one method for analyzing state exclusion in which we attempt to minimize the average probability of error. Alternative formulations were presented in Section \ref{Section: Alternative measures of exclusion}
and it would be interesting to study the relationships between them and that defined in Eq. (\ref{SDP Primal}).

\section*{Acknowledgments}

Part of this work was completed while S.B. and J.O. were visiting the Center for Quantum Technologies, Singapore and while R.J. was visiting the Bose Institute, Kolkata, India. S.B. thanks CQT for their support. The work of R.J. is supported by the Singapore Ministry of Education Tier 3 Grant and also the Core Grants of the Centre for Quantum Technologies, Singapore. J.O. is supported by the Royal Society and an EPSRC Established Career fellowship.

\bibliographystyle{apsrev4-1}

\clearpage

\widetext
\appendix
\appendixpage

This supplementary material contains five sections. In Appendix \ref{State Exclusion SDP Formulation}, we give the general definition of an SDP, derive the dual problem for the state exclusion SDP and show the relation to state discrimination. Next, in Appendix \ref{Strong Duality}, we show that the SDP exhibits strong duality and give the proof of Theorem \ref{SD theorem} from the main text. Appendix \ref{Necessary Conditions} derives the necessary condition for conclusion exclusion to be possible given in Theorem \ref{Necc Cond Theorem} as well as the associated corollary. It also contains the proof of the bound on the error probability of state exclusion, Theorem \ref{lower bound}. The PBR game is analyzed in Appendix \ref{PBR Game}. Finally, in Appendix \ref{Alternative State Exclusion SDPs}, we state alternative state exclusion SDPs.

\section{State Exclusion SDP Formulation} \label{State Exclusion SDP Formulation}
Contains:
\begin{itemize}
\item General definition of an SDP.
\item Derivation of the state exclusion SDP dual.
\item Recasting of state exclusion as a discrimination problem.
\end{itemize}

\subsection{General SDPs} \label{General SDPs}

In this section we state the general form of a Semidefinite Program as given in \cite{Watrous2011}. A semidefinite program is defined by three elements $\{A,B,\Phi\}$. $A$ and $B$ are Hermitian matrices, $A\in \textit{Herm}(\mathcal{X})$ and $B\in \textit{Herm}(\mathcal{Y})$, where $\mathcal{X}$ and $\mathcal{Y}$ are complex Euclidean spaces. $\Phi$ is a Hermicity preserving super-operator which takes elements in $\mathcal{X}$ to elements in $\mathcal{Y}$.

From these three elements, two optimization problems can be defined:
\begin{align}
\begin{split} \label{Prime}
\textrm{Primal Problem}\\
\underset{X}{\textrm{Minimize}} :& \quad \alpha=\tr[AX].\\
\textrm{Subject to} :& \quad \Phi(X)= B,\\
& \quad X\geq0.
\end{split}
\end{align}
\begin{align}
\begin{split} \label{Dual}
\textrm{Dual Problem}\\
\underset{Y}{\textrm{Maximize}} :& \quad \beta=\tr[BY].\\
\textrm{Subject to} :& \quad \Phi^{*}(Y)\leq A,\\
& \quad Y\in \textit{Herm}(\mathcal{Y}).
\end{split}
\end{align}
Here $\Phi^*$ is the dual map to $\Phi$ and is defined by:
\begin{align}
\tr[Y\Phi(X)]=\tr[X\Phi^{*}(Y)]. \label{Phi* equation}
\end{align}
We define the optimal solutions to the primal and dual problems to be $\alpha^{*}=\textrm{inf}_X \alpha$ and $\beta^{*}=\textrm{sup}_Y \beta$ respectively.

\subsection{State Exclusion SDP} \label{State Exclusion SDP}

Looking at the state exclusion primal problem,
 Eq. (\ref{SDP Primal}),
 we see that for the exclusion SDP:
\begin{itemize}
\item $A$ is a  $kd$ by $kd$ block diagonal matrix with each $d$ by $d$ block, labeled by $i$, given by $\tilde{\rho}_i$:
\begin{equation}
A=\left(\begin{array}{ccc}
\tilde{\rho}_1 &  & \\
 & \ddots & \\
& & \tilde{\rho}_k
\end{array}\right).
\end{equation}
\item $B$ is the $d$ by $d$ identity matrix.
\item $X$, the variable matrix, is a $kd$ by $kd$ block diagonal matrix where we label each $d$ by $d$ block diagonal by $M_i$:
\begin{equation}
X=\left(\begin{array}{ccc}
M_1 &  & \\
 & \ddots & \\
& & M_k
\end{array}\right).
\end{equation}
\item $Y$ is the $d$ by $d$ matrix we call $N$.
\item The map $\Phi$ is given by $\Phi(X)=\sum_{i} M_{i}$.
\end{itemize}

Using Eq. (\ref{Phi* equation}) we see that $\Phi^*$ must satisfy:
\begin{align}
\tr\left[N\sum_{i=1}^{k} M_i\right]=\tr\left[\left(\begin{array}{ccc}
M_1 & & \\
& \ddots & \\
& & M_{k}
\end{array} \right)\Phi^{*}(N)\right],
\end{align}
and hence $\Phi^{*}(N)$ produces a $kd$ by $kd$ block diagonal matrix with $N$ in each of the block diagonals:
\begin{equation}
\Phi^{*}(N)=\left(\begin{array}{ccc}
N & & \\
& \ddots & \\
& & N\end{array}\right).
\end{equation}

Substituting these elements into Eq. (\ref{Dual}), we obtain the dual SDP for state exclusion
 as stated in Eq. (\ref{SDP Dual})
.

\subsection{The relation between state discrimination and state exclusion} \label{exc-disc derivation}

Here we give the derivation of Eq. (\ref{exc-disc}).

Given $\mathcal{P}$ we define:
\begin{equation*}
\mathcal{R}=\left\{\vartheta_i=\frac{1}{k-1}\sum_{j\neq i}\tilde{\rho}_j\right\}_{i=1}^{k}.
\end{equation*}
Then, in state discrimination on $\mathcal{R}$ we would attempt to minimize:
\begin{align*}
P_{\textit{error}}^{\textit{dis}}\left(\mathcal{R}\right)&=1-\sum_{i=1}^{k}\tr\left[\vartheta_i M_i\right],\\
&=1-\sum_{i=1}^{k}\sum_{j\neq i}\frac{1}{k-1}\tr\left[\tilde{\rho}_j M_i\right],\\
&=1-\frac{1}{k-1}\sum_{i=1}^{k}\sum_{j=1}^{k}\tr\left[\tilde{\rho}_j M_i\right]+\frac{1}{k-1}\sum_{i=1}^{k}\tr\left[\tilde{\rho}_i M_i\right],\\
&=\frac{k-2}{k-1}+\frac{1}{k-1}P_{\textit{error}}^{\textit{exc}}\left(\mathcal{P}\right).
\end{align*}

\section{Strong Duality} \label{Strong Duality}
Contains:
\begin{itemize}
\item Statement of Slater's Theorem.
\item Proof that the exclusion SDP satisfies the conditions of Slater's Theorem.
\item Derivation of necessary and sufficient conditions for a measurement to be optimal for performing exclusion (proof of Theorem 
\ref{SD theorem}
).
\end{itemize}

\subsection{Slater's Theorem} \label{Slater's Theorem}

Slater's Theorem provides a means to test whether an SDP satisfies strong duality ($\alpha^{*}=\beta^{*}$).

\begin{theorem}{(Slater's Theorem.)}
The following implications hold for every SDP:
\begin{enumerate}
\item If there exists a feasible solution to the primal problem and a Hermitian operator $Y$ for which $\Phi^*(Y)<A$, then $\alpha^{*}=\beta^{*}$ and there exists a feasible $X^*$ for which $\tr[AX^*]=\alpha^{*}$.
\item If there exists a feasible solution to the dual problem and a positive semidefinite operator $X$ for which $\Phi(X)=B$ and $X>0$, then $\alpha^{*}=\beta^{*}$ and there exists a feasible $Y^*$ for which $\tr[BY^*]=\beta^{*}$.
\end{enumerate}
\end{theorem}

\subsection{Slater's Theorem applied to Exclusion SDP} \label{Slater's Theorem applied to Exclusion SDP}

To see that the exclusion SDP satisfies the conditions of Slater's Theorem consider $X=\frac{1}{k} \mathbb{I}$ and $N=-\mathbb{I}$ (where the Identity matrices are taken to have the correct dimension). $X$ is strictly positive definite and so strictly satisfies the constraints of the primal problem. $N<0$ and hence $N<\tilde{\rho}_{i}$, $\forall i$, so $N$ strictly satisfies the constraints of the dual problem.

\subsection{Necessary and sufficient conditions for a measurement to be optimal} \label{Necessary and sufficient conditions for a measurement to be optimal}

To prove Theorem 
\ref{SD theorem}
we will need the following fact about SDPs:
\begin{proposition}{(Complementary Slackness.)} \label{Slackness}
Suppose $X$ and $Y$, which are feasible for the primal and dual problems respectively, satisfy $\tr[AX]=\tr[BY]$. Then it holds that:
\begin{equation*}
\Phi^*(Y)X=AX \textrm{ and }\Phi(X)Y=BY.
\end{equation*}
\end{proposition}

We now give the proof for Theorem 
\ref{SD theorem}
.

\begin{proof}{(Proof of Theorem 
\ref{SD theorem}
.)}

Suppose we are given a valid measurement, $\mathcal{M}=\{M_i\}_{i=1}^{k}$, and that $N$, defined by:
\begin{equation*}
N=\sum_{i=1}^{k} \tilde{\rho}_i M_i,
\end{equation*}
satisfies the constraints of the dual problem. Then:
\begin{align*}
\beta&=\tr[N],\\
&=\tr\left[\sum_{i=1}^{k} \tilde{\rho}_i M_i\right],\\
&=\sum_{i=1}^{k}\tr\left[\tilde{\rho}_i M_i\right],\\
&=\alpha.
\end{align*}
Hence, by strong duality, $\mathcal{M}$ is an optimal measurement.

Now suppose $\mathcal{M}$ is an optimal measurement. By Proposition \ref{Slackness}, an optimal $N$ satisfies:
\begin{align*}
\Phi^{*}(N)\left(\begin{array}{ccc}
M_1 & & \\
& \ddots & \\
& & M_{k}
\end{array} \right)
&=
\left(\begin{array}{ccc}
\tilde{\rho}_1 M_1 & & \\
& \ddots & \\
& & \tilde{\rho}_k M_{k}
\end{array} \right),\\
\Rightarrow\quad\left(\begin{array}{ccc}
N M_1 & & \\
& \ddots & \\
& & N M_{k}
\end{array} \right)
&=
\left(\begin{array}{ccc}
\tilde{\rho}_1 M_1 & & \\
& \ddots & \\
& & \tilde{\rho}_k M_{k}
\end{array} \right),
\end{align*}
which implies that:
\begin{equation*}
N M_i = \tilde{\rho}_i M_i, \quad \forall i.
\end{equation*}
Taking the sum over $i$ on both sides and using the fact that $\sum_i M_i=\mathbb{I}$, we obtain:
\begin{equation*}
N=\sum_{i=1}^{k} \tilde{\rho}_i M_i,
\end{equation*}
as required.
\end{proof}

\section{Necessary Conditions and Bounds} \label{Necessary Conditions}
Contains:
\begin{itemize}
\item Derivation of a necessary condition for conclusive exclusion to be possible (proof of Theorem 
\ref{Necc Cond Theorem}
).
\item Derivation of the necessary condition for conclusive state discrimination to be possible 
.
\item Derivation of the lower bound on the error probability for the exclusion task (proof of Theorem \ref{lower bound}).
\end{itemize}

\subsection{Necessary Condition for Conclusive Exclusion} \label{Necessary Condition for Conclusive Exclusion}

Here we derive the necessary condition for single state conclusive exclusion to be possible that was given in Theorem 
\ref{Necc Cond Theorem}
.

\begin{proof}{(Proof of Theorem 
\ref{Necc Cond Theorem}
.)}

Suppose that $\mathcal{P}=\{\rho_i\}_{i=1}^{k}$. A feasible solution to the dual SDP, $N$, must be Hermitian and satisfy $N\leq\rho_i$, $\forall i$. Our goal is to construct such an $N$ with the property $\tr[N]>0$. If this is possible, conclusive exclusion is not possible.

First we define $U_{jl}$ to be a unitary such that $\tr\left[\sqrt{\rho_l}\sqrt{\rho_j}U_{jl}\right]=F(\rho_j,\rho_l)$ and note that $U_{lj}=U_{jl}^{*}$. We construct $N$ as follows (for $p,\epsilon \in (0,1) $):
\begin{align*}
N=-p\sum_{r=1}^{k}\rho_r +\frac{1-\epsilon}{k-2}p\sum_{1\leq j<l\leq k}\left(\sqrt{\rho_j}U_{jl}\sqrt{\rho_l}+\sqrt{\rho_l}U_{jl}^{*}\sqrt{\rho_j}\right),
\end{align*}
and note that $N$ is Hermitian. Now consider:
\begin{align*}
\rho_1-N&=(1+p)\rho_1+p\sum_{r=2}^{k}\rho_r-\frac{1-\epsilon}{k-2}p\sum_{1\leq j<l\leq k}\left(\sqrt{\rho_j}U_{jl}\sqrt{\rho_l}+\sqrt{\rho_l}U_{jl}^{*}\sqrt{\rho_j}\right),\\
&=\sum_{r=2}^{k}\left[\frac{1+p}{k-1}\rho_1+\epsilon p\rho_r -\frac{1-\epsilon}{k-2}p\left(\sqrt{\rho_1}U_{1r}\sqrt{\rho_r}+\sqrt{\rho_r}U_{1r}^{*}\sqrt{\rho_1}\right)\right]\\
&\quad +\frac{1-\epsilon}{k-2}p\sum_{2\leq j<l\leq k} \left[\rho_j+\rho_l -\sqrt{\rho_j}U_{jl}\sqrt{\rho_l}-\sqrt{\rho_l}U_{jl}^{*}\sqrt{\rho_j}\right],\\
&=\sum_{r=2}^{k}\left[\frac{1+p}{k-1}\rho_1+\epsilon p\rho_r -\frac{1-\epsilon}{k-2}p\left(\sqrt{\rho_1}U_{1r}\sqrt{\rho_r}+\sqrt{\rho_r}U_{1r}^{*}\sqrt{\rho_1}\right)\right]\\
&\quad + \frac{1-\epsilon}{k-2}p\sum_{2\leq j<l\leq k} \left(\sqrt{\rho_j}\sqrt{U_{jl}}-\sqrt{\rho_l}\sqrt{U_{jl}^{*}}\right)\left(\sqrt{U_{jl}^{*}}\sqrt{\rho_j}-\sqrt{U_{jl}}\sqrt{\rho_l}\right).
\end{align*}
The terms in the second summation on the last line are positive semidefinite. Consider, individually, the terms in the first summation:
\begin{align*}
&\frac{1+p}{k-1}\rho_1+\epsilon p\rho_r -\frac{1-\epsilon}{k-2}p\left(\sqrt{\rho_1}U_{1r}\sqrt{\rho_r}+\sqrt{\rho_r}U_{1r}^{*}\sqrt{\rho_1}\right),\\
=&\left[\frac{1+p}{k-1}-\left(\frac{(1-\epsilon)p}{k-2}\right)^2\frac{1}{\epsilon p}\right]\rho_1\\
&\quad+\left[\left(\frac{(1-\epsilon)p}{k-2}\right)^2\frac{1}{\epsilon p}\right]\rho_1 +\epsilon p\rho_r -\frac{1-\epsilon}{k-2}p\left(\sqrt{\rho_1}U_{1r}\sqrt{\rho_r}+\sqrt{\rho_r}U_{1r}^{*}\sqrt{\rho_1}\right),\\
=&\left[\frac{1+p}{k-1}-\left(\frac{(1-\epsilon)p}{k-2}\right)^2\frac{1}{\epsilon p}\right]\rho_1\\
&\quad+\left(\frac{(1-\epsilon)p}{(k-2)\sqrt{\epsilon p}}\sqrt{\rho_1}\sqrt{U_{1r}}-\sqrt{\epsilon p}\sqrt{\rho_r}\sqrt{U_{1r}^{*}}\right)\left(\frac{(1-\epsilon)p}{(k-2)\sqrt{\epsilon p}}\sqrt{U_{1r}^{*}}\sqrt{\rho_1}-\sqrt{\epsilon p}\sqrt{U_{1r}}\sqrt{\rho_r}\right).
\end{align*}
Hence, for $\rho_1-N$ to be positive semidefinite, we need the first term in the last line to be positive:
\begin{align}
\left[\frac{1+p}{k-1}-\left(\frac{(1-\epsilon)p}{k-2}\right)^2\frac{1}{\epsilon p}\right]&\geq 0, \nonumber\\
\frac{\epsilon}{\frac{(k-1)(1-\epsilon)^2}{(k-2)^2}-\epsilon}&\geq p. \label{p epsilon condition}
\end{align}
Therefore, provided $p$ and $\epsilon$ satisfy Eq. (\ref{p epsilon condition}), $N\leq \rho_1$. Similarly, one can argue that $\rho_i\leq N$, $\forall i$ and hence $N$ is a feasible solution to the dual problem.

We now wish to know under what conditions we have $\tr[N]>0$:
\begin{align*}
\begin{array}{crl}
&\tr\left[N\right]>&0,\\
\Rightarrow&-kp+\frac{1-\epsilon}{k-2}p\sum_{1\leq j<l\leq k}\tr\left[\sqrt{\rho_j}U_{jl}\sqrt{\rho_l}+\sqrt{\rho_l}U_{jl}^{*}\sqrt{\rho_j}\right]>&0,\\
\Rightarrow&\sum_{j\neq l=1}^{k} F(\rho_j,\rho_l)>&\frac{k(k-2)}{1-\epsilon}.
\end{array}
\end{align*}
Letting $\epsilon\rightarrow 0$ and using weak duality we obtain our result. Conclusive exclusion is not possible if $\sum_{j\neq l=1}^{k} F(\rho_j,\rho_l)>k(k-2)$.
\end{proof}

\subsection{Necessary condition for conclusive state discrimination} \label{Necessary condition for conclusive state discrimination}

Here we show how the necessary condition for perfect state discrimination to be possible
 can be derived from our necessary condition on conclusive state exclusion, Theorem 
\ref{Necc Cond Theorem}
.

\begin{corollary} \label{State Discrimination}
Conclusive state discrimination on the set $\mathcal{P}=\{\rho_i\}_{i=1}^{k}$ is possible only if $\mathcal{P}$ is an orthogonal set.
\end{corollary}

\begin{proof}{(Proof of Corollary 
\ref{State Discrimination}
.)}
For $\mathcal{P}=\{\rho_i\}_{i=1}^{k}$, define:
\begin{equation*}
\hat{\rho}_j=\frac{1}{k-1}\sum_{i\neq j} \rho_i.
\end{equation*}
Let $j\neq l$ and consider:
\begin{equation*}
A=\frac{1}{k-1}\sum_{r\neq j,l}\rho_r.
\end{equation*}
We first show that $F(\hat{\rho}_j,\hat{\rho}_l)\geq F(\hat{\rho}_j,A)$. Consider:
\begin{align*}
F(\hat{\rho}_j,A)&=\tr\left[\sqrt{\sqrt{\hat{\rho}_j}A\sqrt{\hat{\rho}_j }}\right],\\
&\leq\tr\left[\sqrt{\sqrt{\hat{\rho}_j}\hat{\rho}_l\sqrt{\hat{\rho}_j }}\right],\\
&=F(\hat{\rho}_j,\hat{\rho}_l).
\end{align*}
The inequality follows from the following facts:
\begin{enumerate}
\item It can be easily seen from the definitions that $A\leq\hat{\rho}_l$.
\item If $B\geq C$ then $D^{*}BD\geq D^{*}CD$, $\forall D$. Hence:
\begin{align*}
\sqrt{\hat{\rho}_j}A\sqrt{\hat{\rho}_j}\leq\sqrt{\hat{\rho}_j}\hat{\rho}_l\sqrt{\hat{\rho}_j}.
\end{align*}
\item The square root function is operator monotone, so:
\begin{align*}
\sqrt{\sqrt{\hat{\rho}_j}A\sqrt{\hat{\rho}_j}}\leq\sqrt{\sqrt{\hat{\rho}_j}\hat{\rho}_l\sqrt{\hat{\rho}_j}}.
\end{align*}
\item The trace function is operator monotone and so finally:
\begin{align*}
\tr\left[\sqrt{\sqrt{\hat{\rho}_j}A\sqrt{\hat{\rho}_j}}\right]\leq\tr\left[\sqrt{\sqrt{\hat{\rho}_j}\hat{\rho}_l\sqrt{\hat{\rho}_j}}\right].
\end{align*}
\end{enumerate}

Using a similar argument to the above, it is possible to show that:
\begin{equation*}
F(\hat{\rho}_j,A)\geq F(A,A) =\frac{k-2}{k-1}.
\end{equation*}

If ${\rho}_j$, ${\rho}_l$ and $A$ are pairwise orthogonal, then $\hat{\rho}_j$ and $\hat{\rho}_l$ commute and are simultaneously diagonalizable. This means that:
\begin{align*}
F(\hat{\rho}_j,\hat{\rho}_l)&=\left|\left|\sqrt{\hat{\rho}_j}\sqrt{\hat{\rho}_l}\right|\right|_{\textit{Tr}},\\
&=\left|\left|A\right|\right|_{\textit{Tr}},\\
&=F(A,A),\\
&=\frac{k-2}{k-1}.
\end{align*}

Now suppose that $\rho_j$ and $A$ are not orthogonal. We take $\{a_r\}$ to be the eigenvalues and $\{\ket{v_r}\}$ to be the eigenvectors of $\sqrt{A}$, so:
\begin{align*}
F(\hat{\rho}_l,A)&\geq\tr\left[\sqrt{\hat{\rho_l}}\sqrt{A}\right],\\
&=\sum_{r}a_r\bra{v_r}\sqrt{\hat{\rho_l}}\ket{v_r}.
\end{align*}
We know that $\sqrt{\hat{\rho}_l}\geq\sqrt{A}$ and hence:
\begin{equation*}
\bra{v_r}\sqrt{\hat{\rho}_l}\ket{v_r}\geq a_r,\quad \forall r.
\end{equation*}
As $\rho_j$ and $A$ are not orthogonal:
\begin{align*}
\sum_r\bra{v_r}\sqrt{\hat{\rho}_l}\ket{v_r}>\sum_r a_r,
\end{align*}
and there must exist some $r$ such that:
\begin{equation*}
\bra{v_r}\sqrt{\hat{\rho}_l}\ket{v_r}> a_r.
\end{equation*}
Hence:
\begin{align*}
F(\hat{\rho}_l,A)&\geq\sum_{r}a_r\bra{v_r}\sqrt{\hat{\rho_l}}\ket{v_r},\\
&>\sum_r a_{r}^{2},\\
&=\tr\left[A\right], \\
&=\frac{k-2}{k-1}.
\end{align*}
So $F(\hat{\rho}_j,\hat{\rho}_l)=(k-2)/(k-1)$, $\forall l\neq j$, if and only if $\mathcal{P}$ is an orthogonal set.

By Theorem 
\ref{Necc Cond Theorem}
, for conclusive $(m-1)$-state exclusion (and hence conclusive state discrimination) to be possible, we require that:
\begin{align*}
\sum_{j\neq l=1}^{k} F(\hat{\rho}_j,\hat{\rho}_l)=k(k-2),
\end{align*}
which implies that $\mathcal{P}$ must be an orthogonal set.
\end{proof}

\subsection{Bound on success probability} \label{lower bound proof}

In this section we give the proof of Theorem \ref{lower bound}.

\begin{proof}{(Proof of Theorem \ref{lower bound}.)}
The goal is to show that $N_{\varepsilon}\leq\tilde{\rho}_i$, $\forall i$, where $N_{\varepsilon}$ is defined in Eq. (\ref{Nepsilon}). Recall that given two Hermitian operators, $A$ and $B$, $\min\left(A,B\right)$ is defined by:
\begin{equation}
\min\left(A,B\right)=\frac{1}{2}\left[A+B-\left|A-B\right|\right].
\end{equation}
Note that $\min(A,B)\leq A$ and $\min(A,B)\leq B$ as:
\begin{align*}
A-\min(A,B)&=\frac{1}{2}\left[A-B+\left|A-B\right|\right],\\
&=\frac{1}{2}\left[\sum_{i=1}^{d}\lambda_i\ketbra{u_i}{u_i}+\sum_{i=1}^{d}\left|\lambda_i\right|\ketbra{u_i}{u_i}\right],\\
&\geq 0,
\end{align*}
and similarly $B-\min(A,B)\geq 0$. Here $\sum_{i=1}^{d}\lambda_i\ketbra{u_i}{u_i}$ is the spectral decomposition of $A-B$.

The bound is obtained by constructing $N_{\varepsilon}$ iteratively as follows:
\begin{align*}
N_{\varepsilon}^{(2)}&=\min\left(\tilde{\rho}_{\varepsilon(2)},\tilde{\rho}_{\varepsilon(1)}\right),\\
N_{\varepsilon}^{(3)}&=\min\left(\tilde{\rho}_{\varepsilon(3)},N_{\varepsilon}^{(2)}\right),\\
\vdots&=\vdots\\
N_{\varepsilon}=N_{\varepsilon}^{(k)}&=\min\left(\tilde{\rho}_{\varepsilon(k)},N_{\varepsilon}^{(k-1)}\right).
\end{align*}
Using the fact that $\min(A,B)\leq A$ and $\min(A,B)\leq B$, by construction we have $N_\varepsilon\leq\tilde{\rho}_i$, $\forall i$. 
\end{proof}

\section{PBR Game} \label{PBR Game}
Contains:
\begin{itemize}
\item Proof that the set of projectors $\mathcal{M}=\{\ket{\zeta_{\vec{x}}}\}_{\vec{x}\in\{0,1\}^n}$, as given in 
Eq. (\ref{PBR Projector})
, forms a valid measurement.
\item Derivation of the conditions under which $\mathcal{M}$ is the optimal measurement for performing exclusion in the PBR game.
\item Derivation of how well $\mathcal{M}$ performs at the exclusion task.
\end{itemize}

\subsection{Proof that $\mathcal{M}$ is a measurement} \label{Proof that M is a measurement}
To see that $\mathcal{M}=\{\ket{\zeta_{\vec{x}}}\}_{\vec{x}\in\{0,1\}^n}$, where:
\begin{equation}
\ket{\zeta_{\vec{x}}}=\frac{1}{\sqrt{2^n}}\left(\ket{\vec{0}}-\sum_{\vec{r}\neq\vec{0}}\left(-1\right)^{\vec{x}\cdot\vec{r}}\ket{\vec{r}}\right),
\end{equation}
forms a valid measurement we shall show that it is a set of orthogonal vectors. Consider:
\begin{align*}
\bracket{\zeta_{\vec{s}}}{\zeta_{\vec{t}}}&=\frac{1}{2^n}\left(\bra{\vec{0}}-\sum_{\vec{r}\neq\vec{0}}\left(-1\right)^{\vec{s}\cdot\vec{r}}\bra{\vec{r}}\right)\left(\ket{\vec{0}}-\sum_{\vec{q}\neq\vec{0}}\left(-1\right)^{\vec{t}\cdot\vec{q}}\ket{\vec{q}}\right),\\
&=\frac{1}{2^n}\left(1+\sum_{\vec{r},\vec{q}\neq\vec{0}}\left(-1\right)^{\vec{s}\cdot\vec{r}}\left(-1\right)^{\vec{t}\cdot\vec{q}}\bracket{\vec{r}}{\vec{q}}\right),\\
&=\frac{1}{2^n}\sum_{\vec{r}}\left(-1\right)^{\left(\vec{s}+\vec{t}\right)\cdot\vec{r}},\\
&=\delta_{\vec{s}\vec{t}}.
\end{align*}
Hence $\mathcal{M}$ is a set of orthogonal vectors and therefore a valid measurement basis.

\subsection{Derivation of conditions under which $\mathcal{M}$ is an optimal measurement} \label{Derivation of conditions under which M is an optimal measurement} 
To show that this measurement, $\mathcal{M}$, is optimal for certain pairs of $n$ and $\theta$, we need to construct an $N$ as per 
Eq. (\ref{N construct}) 
and show that it satisfies the constraints of the dual problem. Writing $\tilde{\rho}_{\vec{x}}=\frac{1}{2^n}\ketbra{\Psi_{\vec{x}}}{\Psi_{\vec{x}}}$ and $M_{\vec{x}}=\ketbra{\zeta_{\vec{x}}}{\zeta_{\vec{x}}}$, we have:
\begin{equation}
N=\frac{1}{2^n}\sum_{\vec{x}}\ket{\Psi_{\vec{x}}}\bracket{\Psi_{\vec{x}}}{\zeta_{\vec{x}}}\bra{\zeta_{\vec{x}}}.
\end{equation}
Note that:
\begin{align*} \langle\Psi_{\vec{x}}|\zeta_{\vec{x}}\rangle&=\frac{1}{{\sqrt{2^n}}}\left(\left[\cos\left(\frac{\theta}{2}\right)\right]^n-\sum^{n}_{i=1}{n\choose{i}}\left[\cos\left(\frac{\theta}{2}\right)\right]^{n-i}\left[\sin\left(\frac{\theta}{2}\right)\right]^i \right),\\
&=\frac{1}{{\sqrt{2^n}}}\left[\cos\left(\frac{\theta}{2}\right)\right]^n\left(2-\left[1+\tan\left(\frac{\theta}{2}\right)\right]^n\right).
\end{align*}
So we have:
\begin{align}
N=C\left(\theta\right)\left[\ketbra{\vec{0}}{\vec{0}}-\sum_{\vec{r}\neq\vec{0}}\left[\tan\left(\frac{\theta}{2}\right)\right]^{|\vec{r}|}\ketbra{\vec{r}}{\vec{r}}\right], \label{PBR N}
\end{align}
where $C(\theta)$ is given by:
\begin{equation}
C\left(\theta\right)=\frac{1}{2^n}\left[\cos\left(\frac{\theta}{2}\right)\right]^{2n}\left(2-\left[1+\tan\left(\frac{\theta}{2}\right)\right]^n\right). \label{C(theta)}
\end{equation}
Note also that $N$ is a real, diagonal matrix and hence is Hermitian so it remains to determine under what conditions $\rho_i-N$ is a positive semidefinite matrix for all $i$. 

Let us define the matrices $A_i$ by:
\begin{equation*}
A_i=-N+\rho_i.
\end{equation*}
The goal is to prove that none of the $A_i$ have a negative eigenvalue. Say $A_i$ has eigenvalues $\{a^{r}_{i}\}$ where $a^{1}_{i}\geq a^{2}_{i}\geq\ldots a^{2^n}_{i}$. The matrix $-N$ has eigenvalues $\{v^{r}\}$ where for $1\leq r\leq 2^n-1$:
\begin{align}
v^{r}&=C\left(\theta\right)\left[\tan\left(\frac{\theta}{2}\right)\right]^{|\vec{r}|},
\end{align}
and for $r=2^n$:
\begin{align}
v^{2^n}&=-C\left(\theta\right).
\end{align}
Each $\rho_i$ is a rank 1 density matrix and hence have eigenvalues $u^{1}_{i}=1$ and $u^{r}_{i}=0$ for $2\leq r\leq 2^n$.

By Weyl's inequality:
\begin{align}
v^{r}+u^{2^n}_{i}\leq a^{r}_{i}.
\end{align}
So, provided $C(\theta)>0$, we have $a^{r}_{i}>0$ for $1\leq r\leq2^n-1$. Hence at most one eigenvalue of $A_i$ is non-positive. Investigating this non-positive eigenvalue further, consider $A_i$ acting on the state $\ket{\zeta_i}$:
\begin{align*}
A_i\ket{\zeta_i}&=\rho_i\ket{\zeta_i}-\sum_{j=1}^{2^n}\rho_j\ket{\zeta_j}\bracket{\zeta_j}{\zeta_i}, \\
&=0.
\end{align*}
Hence the non-positive eigenvalue of $A_i$ is 0 implying that $A_i\geq 0$, $\forall i$, which in turn implies that $N\leq\rho_i$, $\forall i$, provided $C(\theta)>0$. As $\left[\cos\left(\theta/2\right)\right]^{2n}\geq 0$, we have shown that $\{\ket{\zeta_{\vec{x}}}\}_{\vec{x}\in\{0,1\}^n}$
, as defined in  Eq. (\ref{PBR Projector})
, is the optimal measurement for exclusion provided:
\begin{equation}
\left(2-\left[1+\tan\left(\frac{\theta}{2}\right)\right]^n\right)>0. \label{PBR Criterion Complement}
\end{equation}
This region is the complement of that given 
in Eq. (\ref{PBR Criterion}) 
so we know the optimal measurement to perform for all values of $n$ and $\theta$.

\subsection{Derivation of how well $\mathcal{M}$ performs at the exclusion task} \label{Derivation of how well M performs at the exclusion task}

Is conclusive exclusion possible in the region defined by Eq. (\ref{PBR Criterion Complement})? To answer this we must consider the trace of the $N$ given in Eq. (\ref{PBR N}):
\begin{equation}
\tr[N]=\frac{1}{2^n}\left[\cos\left(\frac{\theta}{2}\right)\right]^{2n}\left(2-\left[1+\tan\left(\frac{\theta}{2}\right)\right]^n\right)^2.
\end{equation}
This is strictly positive provided and hence conclusive exclusion is not possible. $\tr[N]$ does however, tell us how accurately we can perform state exclusion when we can not do it conclusively.

\section{Alternative State Exclusion SDPs} \label{Alternative State Exclusion SDPs}

Contains:
\begin{itemize}
\item Derivation of unambiguous state exclusion SDP dual.
\item Derivation of worst case error SDP dual.
\end{itemize}

\subsection{Unambiguous State Exclusion SDP} \label{Unambiguous State Exclusion SDP}

In this section the dual problem for the primal SDP for unambiguous state exclusion as given in Eq. (\ref{Unambig Prime}) is derived.

Comparing Eq. (\ref{Unambig Prime}) with Eq. (\ref{Prime}), we see that here:

\begin{itemize}
\item $A$ is a  $kd$ by $kd$ block diagonal matrix with each $d$ by $d$ block containing $\sum_{j=1}^{k}\tilde{\rho}_j$:
\begin{equation}
A=\left(\begin{array}{ccc}
\sum_{j=1}^{k}\tilde{\rho}_j &  & \\
 & \ddots & \\
& & \sum_{j=1}^{k}\tilde{\rho}_j
\end{array}\right).
\end{equation}
\item $B$ is a $(d+k)$ by $(d+k)$ matrix with the top left $d$ by $d$ block being an identity matrix and all other elements being $0$:
\begin{equation}
B=\left(\begin{array}{cc}
\mathbb{I} & 0\\
0 & 0
\end{array}\right).
\end{equation}
\item $X$, the variable matrix, is a $kd$ by $kd$ block diagonal matrix where we label each $d$ by $d$ block diagonal by $M_i$:
\begin{equation}
X=\left(\begin{array}{ccc}
M_1 &  & \\
 & \ddots & \\
& & M_k
\end{array}\right).
\end{equation}
\item $Y$ is a $(d+k)$ by $(d+k)$ matrix whose top left $d$ by $d$ block we call $N$ and the remaining $k$ diagonal elements we label by $a_i$.
\begin{equation}
Y=\left(\begin{array}{cccc}
N & & &\\
& a_1 & &\\
& & \ddots & \\
& & & a_k
\end{array}\right).
\end{equation}
\item The map $\Phi$ is given by:
\begin{equation}
\Phi(X)=\left(\begin{array}{cccc}
\sum_{i=1}^{k}M_i & & &\\
& \tr\left[\tilde{\rho}_1 M_1\right] & &\\
& & \ddots &\\
& & & \tr\left[\tilde{\rho}_k M_k\right]
\end{array}\right).
\end{equation}
\end{itemize}

Using Eq. (\ref{Phi* equation}) we see that $\Phi^*$ must satisfy:
\begin{align}
\tr\left[N\sum_{i=1}^{k} M_i\right]+\sum_{i=1}^{k}a_i\tr\left[\tilde{\rho}_i M_i\right]=\tr\left[\left(\begin{array}{ccc}
M_1 & & \\
& \ddots & \\
& & M_{k}
\end{array} \right)
\Phi^{*}\left[\left(\begin{array}{cccc}
N & & &\\
& a_1 & &\\
& & \ddots & \\
& & & a_k
\end{array}\right)\right]\right],
\end{align}
and hence $\Phi^{*}(Y)$ produces a $kd$ by $kd$ block diagonal matrix:
\begin{equation}
\Phi^{*}(Y)=\left(\begin{array}{ccc}
N+a_1 \tilde{\rho}_1 & & \\
& \ddots & \\
& & N+a_k \tilde{\rho}_k\end{array}\right).
\end{equation}

Substituting these elements into Eq. (\ref{Dual}) and taking into account the fact that we are maximizing rather than minimizing in the primal problem, we obtain the dual SDP as stated in Eq. (\ref{Unambig Dual}).

\subsection{Worst Case Error State Exclusion SDP} \label{Worst Case Error State Exclusion SDP}

In this section the dual problem for the primal SDP for worst case error state exclusion as given in Eq. (\ref{Worst Case Prime}) is derived.

Comparing Eq. (\ref{Worst Case Prime}) with Eq. (\ref{Prime}), we see that here:

\begin{itemize}
\item $A$ is a  $(kd+1)$ by $(kd+1)$ matrix with $A_{11}=1$ being the only non-zero element:
\begin{equation}
A=\left(\begin{array}{cccc}
1 &  & &\\
 & 0 & & \\
& & \ddots &\\
& & & 0
\end{array}\right).
\end{equation}
\item $B$ is a $(d+k)$ by $(d+k)$ where the bottom right $d$ by $d$ block is the identity matrix. All other elements are zero:
\begin{equation}
B=\left(\begin{array}{cc}
0 & 0\\
0 & \mathbb{I}
\end{array}\right).
\end{equation}
\item $X$, the variable matrix, is a $kd+1$ by $kd+1$ block diagonal matrix where $X_{11}=\lambda$ and we label each subsequent $d$ by $d$ block diagonal by $M_i$:
\begin{equation}
X=\left(\begin{array}{cccc}
\lambda & & &\\
& M_1 &  & \\
& & \ddots & \\
& & & M_k
\end{array}\right).
\end{equation}
\item $Y$ is a $(d+k)$ by $(d+k)$ matrix whose bottom right $d$ by $d$ block we call $N$ and the remaining $k$ diagonal elements we label by $a_i$.
\begin{equation}
Y=\left(\begin{array}{cccc}
a_1 & & &\\
&  \ddots & & \\
& &  a_k &\\
& & & N
\end{array}\right).
\end{equation}
\item The map $\Phi$ is given by:
\begin{equation}
\Phi(X)=\left(\begin{array}{cccc}
\lambda-\tr\left[\tilde{\rho}_1 M_1\right] & &\\
& \ddots & &\\
& & \lambda-\tr\left[\tilde{\rho}_k M_k\right] &\\
& & & \sum_{i=1}^{k}M_i
\end{array}\right).
\end{equation}
\end{itemize}

Using Eq. (\ref{Phi* equation}) we see that $\Phi^*$ must satisfy:
\begin{align}
\lambda\sum_{i=1}^{k} a_i - \sum_{i=1}^{k} a_i\tr\left[\tilde{\rho}_i M_i\right]=\tr\left[\left(\begin{array}{cccc}
\lambda & & &\\
& M_1 &  & \\
& & \ddots & \\
& & & M_k
\end{array}\right)
\Phi^{*}\left[\left(\begin{array}{cccc}
a_1 & & &\\
&  \ddots & & \\
& &  a_k &\\
& & & N
\end{array}\right)\right]\right],
\end{align}
and hence $\Phi^{*}(Y)$ produces a $kd$ by $kd$ block diagonal matrix:
\begin{equation}
\Phi^{*}(Y)=\left(\begin{array}{cccc}
\sum_{i=1}^{k} a_i & & &\\
& N-a_1 \tilde{\rho}_1 & & \\
& & \ddots & \\
& & & N-a_k \tilde{\rho}_k\end{array}\right).
\end{equation}

Substituting these elements into Eq. (\ref{Dual}), we obtain the dual SDP as stated in Eq. (\ref{Worst Case Dual}).

\end{document}